\documentclass[reprint,pra,amsmath,amssymb, twocolumn,superscriptaddress, raggedbottom]{revtex4-2}

\usepackage{graphicx}
\usepackage{dcolumn}
\usepackage{bm}
\usepackage{bbm}
\usepackage[colorlinks=true, urlcolor=blue,citecolor=blue,anchorcolor=blue,breaklinks=true]{hyperref}
\usepackage[capitalise]{cleveref}
\usepackage{dsfont}
\usepackage{mathtools}
\usepackage{amsthm, thmtools}
\usepackage{amsmath}
\usepackage{empheq}
\usepackage{xcolor}
\usepackage[normalem]{ulem} 
\usepackage{comment}

\newcommand{\microspace}{\mspace{0.5mu}}

\declaretheorem[
]{corollary}
\declaretheorem[
]{proposition}

\def\<{\langle}
\def\>{\rangle}
\def \lket {\left|}
\def \rket {\right\rangle}
\def \lbra {\left\langle}
\def \rbra {\right|}
\newcommand{\ket}[1]{\lket\microspace #1 \microspace\rket}

\newcommand{\bra}[1]{\lbra\microspace #1 \microspace\rbra}

\newcommand{\braket}[2]{\langle #1 \microspace | \microspace#2 \rangle}

\newcommand{\avg}[1]{\langle#1\rangle}

\usepackage{mathtools}

\usepackage{bm}
\let\vec\bm 

\begin{document}
\title{Iterative Quantum Algorithms for Maximum Independent Set\\\small{A Tale of Low-Depth Quantum Algorithms}}

\author{Lucas T.~Brady}
\email{lucas.t.brady@nasa.gov}
\affiliation{Quantum Artificial Intelligence Laboratory, NASA Ames Research Center, Moffett Field, California 94035, USA}
\affiliation{KBR, 601 Jefferson St., Houston, TX 77002, USA}
\author{Stuart Hadfield}
\affiliation{Quantum Artificial Intelligence Laboratory, NASA Ames Research Center, Moffett Field, California 94035, USA}
\affiliation{USRA Research Institute for Advanced Computer Science (RIACS), Mountain View, CA 94043, USA}

\date{\today}
\begin{abstract}
Quantum algorithms have been widely studied in the context of combinatorial optimization problems. While this endeavor can often analytically and practically achieve quadratic speedups, theoretical and numeric studies remain limited, especially compared to the study of classical algorithms. We propose and study a new class of hybrid approaches to quantum optimization, termed Iterative Quantum Algorithms, which in particular generalizes the Recursive Quantum Approximate Optimization Algorithm.  This paradigm can incorporate hard problem constraints, which we demonstrate by considering the Maximum Independent Set (MIS) problem. We show that, for QAOA with depth $p=1$, this algorithm performs exactly the same operations and selections as the classical greedy algorithm for MIS.  We then turn to deeper $p>1$ circuits and other ways to modify the quantum algorithm that can no longer be easily mimicked by classical algorithms, and empirically confirm improved performance. Our work demonstrates the practical importance of incorporating proven classical techniques into more effective hybrid quantum-classical algorithms.
\end{abstract}

\maketitle

\section{Introduction} 
Quantum algorithms have been proposed for tackling a wide range of challenging problems, including combinatorial optimization \cite{Kadowaki1998,Farhi2000,Farhi2014,hadfield2019quantum}. In the eagerness to apply quantum approaches to new problems, arguments for the ultimate goal of quantum advantage are often underdeveloped, with much work demonstrating the amenability of quantum approaches to a problem, without exploring whether and why the quantum approach might be advantageous in a given setting.

Analysis of both classical and quantum algorithms is often challenging, but in the case of hybrid quantum algorithms where a quantum processor is tied in with a classical computer that work together in a nontrivial way to solve a problem \cite{Farhi2014,Peruzzo2014,Cerezo2021,bharti2022noisy}, some illuminating reasoning and analysis can be done.  Often, in analyzing these hybrid algorithms, the classical component isn't treated fully, ignoring tasks that are themselves potentially NP-Hard \cite{Bittel2021} or, as in the case discussed in this paper, ignoring the sheer power in the classical components on their own. 
Furthermore, these aspects point to the potential of using the quantum device to enhance existing effective classical approaches.

We consider a broad class of hybrid quantum approaches, which we term Iterative Quantum Algorithms (IQA), exemplified by the recently proposed Recursive Quantum Approximate Optimization Algorithm (RQAOA) \cite{Bravyi2020} and Iterative Quantum Optimization~(IQO)~\cite{dupont2023quantum,hadfield2023iterative} more generally. 
These hybrid approaches tackle a problem iteratively in an interactive fashion, with each iteration reducing the size or connectivity of the problem until the optimal solution is trivial to obtain, which is then unwound to produce an approximate solution to the input problem. 
Here, instead of trying to solve the input problem directly, these approaches utilize the quantum computer to help decide on the particular reduction at each step, which can be viewed as a slightly weaker computational task.

Similar iterative classical algorithms and heuristics, especially greedy approaches, have long been analyzed and applied with great success in the optimization community for a wide variety of problems~\cite{halldorsson1994greed,ausiello2012complexity}. 
An important question is when and to what degree quantum resources can enhance these algorithms~\cite{dupont2023quantum} in hybrid settings -- can we design hybrid IQO algorithms that retain classical performance guarantees while meaningfully incorporating near-term quantum devices?

In this paper, we focus on IQO approaches for Maximum Independent Set (MIS), the problem of finding the largest set of vertices in a graph that does not have any edges between elements of the set.  
MIS is a prototypical NP-hard constrained optimization problem with numerous applications in science and industry~\cite{agarwal1998label,butenko2003maximum,vahdatpour2008theoretical,wurtz2022industry}. 
We review some well-established classical iterative algorithms and then introduce a novel but natural iterative quantum algorithm for MIS based on the Quantum Approximate Optimization Algorithm (QAOA) \cite{Farhi2014,hadfield2019quantum}. 
In the course of our analysis, we show that simplest realization using depth $p=1$ QAOA only has access to the exact same information as the classical algorithm, and then numerically demonstrate that it always uses that information in the exact same way as the classical algorithm.  Therefore, in this simplest setting, there is no quantum distinction, much less advantage \footnote{Indeed, for QAOA with $p=1$, analytic formulas for observable expectation values are obtainable~\cite{wang2018quantum,hadfield2018quantum,ozaeta2022expectation, marwaha2022bounds}. }.

On the other hand, the iterative quantum approach and corresponding classical algorithm then perform comparably, which is much better than if the quantum algorithm had been employed non-iteratively, where performance is much closer to that expected of random guessing. Our results demonstrate that 
utilizing quantum devices in novel ways may facilitate squeezing much higher performance from fixed real-world devices in the near-term than non-iterative approaches would otherwise allow. 

\begin{figure}
\includegraphics[width = 0.48\textwidth]{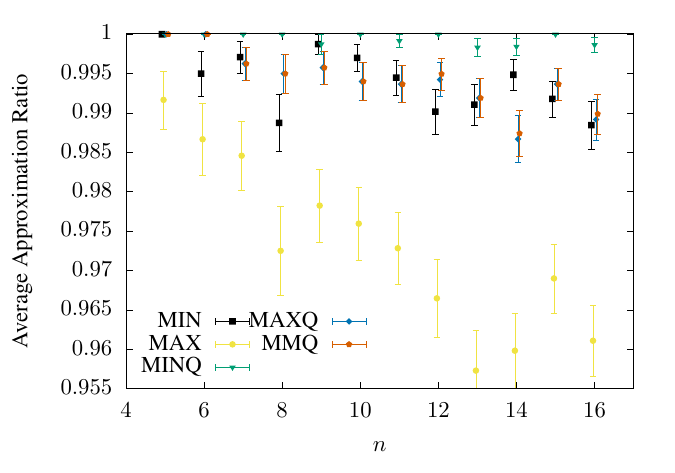}
\caption{The approximation ratio (achieved set size divided by the true maximum independent set size) for various models as a function of $n$ for Erd\"os-R\'enyi graphs with edge probability $q = 1.2\ln(n)/n$.  Each data point is averaged over two hundred randomly chosen graphs, and the error bars show the standard error of the mean.  The ``Q'' algorithms are all performed with $p=2$ QAOA as the quantum component, and ``MINQ'' and ``MAXQ'' refer to the quantum algorithms that reduce to the corresponding classical greedy algorithm at $p=1$.  ``MMQ'' is an IQO algorithm that chooses whether to fix bits on or off at each step as described in Section~\ref{ssec:onoff}. The non-iterative QAOA performance is too low to appear on this plot with the approximation ratio ranging from $0.69$ at $n=5$ down to $0.495$ at $n=16$.  The offsets along the $x$-axis are for visual clarity; all $n$ values are integers.}
\label{fig:p=2_unweighted}
\end{figure}

In the rest of the paper, we 
consider classical greedy algorithms for MIS in Section \ref{sec:classical} and then develop and analyze related iterative quantum approaches in Section~\ref{sec:general_theory}.  We then analyze this basic algorithm for depth $p=1$ in Section~\ref{sec:p=1} before considering higher depth in Section~\ref{sec:higherdepth}.  These quantum algorithms can be further expanded and improved as we discuss in Section~\ref{sec:modifications}.  These expanded and higher depth quantum algorithms do outperform the classical greedy algorithms, but at the same time, they have access to additional resources that could potentially be incorporated into more advanced classical algorithms.  
We then turn to requirements of the quantum devices toward achieving a bonafide quantum advantage.
In Section \ref{sec:higherdepth}, we also provide analytic reasons to believe that on average $n$ vertex graphs, QAOA with depth only $\mathcal{O}(\log n)$ might have enough power for Iterative Quantum Algorithms applied to MIS to produce advantageous approximation ratios.  While our results do suggest paths forward to more effective quantum algorithms, we look forward to more comprehensive comparisons to classical methods as less noisy and larger-scale quantum systems become available.

\paragraph*{Related work:}
Several distinct approaches related to QAOA have been proposed for the Maximum Independent Set problem, differentiated in particular by different means of dealing with the problem's hard constraints~\footnote{Approaches for MIS not directly related to QAOA have also been proposed such as Ref.~\cite{yu2021quantum,djidjev2018efficient}.}.
QAOA approaches based on including penalty terms were previously considered in Refs.~\cite{farhi2020quantum_t,farhi2020quantum_w,saleem2023approaches,cain2022qaoa}. 
A related approach that utilizes Rydberg blockade effects to suppress transitions to infeasible states 
has seen promising development and early experimental implementation on atomic array devices~\cite{pichler2018computational,pichler2018quantum,ebadi2022quantum,wild2021quantum,byun2022finding,kim2022rydberg,nguyen2023quantum,lanthaler2023rydberg,wurtz2022industry,cain2023quantum}. 
Finally, an alternative general approach that modifies the QAOA operators and initial state to only explore the subspace of feasible states was proposed in Ref.~\cite{hadfield2019quantum} and further explored in Refs.~\cite{saleem2020max,saleem2021quantum,tomesh2022quantum,tomesh2022quantumc,saleem2023approaches}.

Whereas these approaches typically solve the problem directly, we employ an iterative approach that utilizes the quantum device to help guide the selections made by a classical problem reduction algorithm. 
Clearly, different problem reduction rules are possible, and IQO is by no means limited to using QAOA. 
We describe iterative algorithms beyond the quantum gate model in Sec.~\ref{ssec:annealing}. 
As mentioned, our protocol generalizes RQAOA~\cite{Bravyi2020,bravyi2020hybrid,bravyi2021classical} as well as a number of preceding classical algorithms, and is closely related to a number of recently proposed iterative schemes~\cite{
hadfield2023iterative,Patel2022,ayanzadeh2022quantum,Bae2023,wagner2023enhancing,dupont2023quantum}. 
We provide details and analysis of the general Iterative Quantum Algorithm framework in a companion paper~\cite{hadfield2023iterative,IQO}, including a wider variety of selection and reduction rules, and applications to different classes of constrained optimization problems. 

During the preparation of this paper, a preprint appeared also considering an Iterative Quantum Algorithm for Maximum Independent Set, 
which the authors called Quantum-Informed Recursive Optimization (QIRO) \cite{Finzgar2023}.  This algorithm is similar to our proposed iterative algorithms, most resembling our MMQ algorithm below but considering both one and two qubit correlators instead of just single qubit expectation values.  
Our work serves to provide expanded theoretical backing behind these approaches, especially with regard to the relation between 
IQA and classical greedy algorithms, and particularly with a view beyond QAOA with depth $p=1$.  
In Ref.~\cite{Finzgar2023} the authors primarily considered $p=1$ QAOA as the quantum subroutine (though suggested that higher $p$ would be advantageous), as well as a quasi-adiabtic evolution on neutral-atom quantum hardware.
An important difference is that our MINQ algorithm is theoretically guaranteed to perform as well as classical MIN, a bar that the QIRO algorithm is not known to meet. 
Below we provide reasoning for why QAOA with $p=1$ is a poor choice as the quantum sub-routine for iterative algorithms in terms of obtaining a quantum advantage and explore higher depth alternatives.

\section{Classical Greedy Algorithms}
\label{sec:classical}
Given a graph $G=(V,E)$ with $|V|=n$ vertices, an independent set is a subset of the vertices with no graph edge between any pair. 
The Max Independent Set problem seeks to determine the largest cardinality independent set in $G$. MIS is a well-known NP-hard problem, and remains hard to approximate, even for bounded degree  graphs~\cite{bazgan2005completeness}.

Classical greedy algorithms are known to perform surprising well~\cite{halldorsson1994greed,schuetz2022combinatorial,angelini2023modern,schuetz2023reply} on MIS problems. 
There are numerous greedy algorithms for this problem (and many other types of heuristics), but one of the most common is called MIN based on the selection rule employed:
\begin{enumerate}
    \item Choose a vertex in the graph with the lowest degree.
    \item Add this vertex to your solution set.
    \item Delete the vertex and its neighbors from the graph.
    \item Repeat steps 1--3 until the graph is empty.
\end{enumerate}

An alternative to this would be MAX which instead deletes the highest degree vertex at each step until a completely disconnected graph remains, providing the solution.  We primarily compare to MIN since it is the better performing algorithm. Nevertheless 
as we will see both are relevant points of comparison when analyzing iterative quantum algorithms.

Performance bounds for the MIN algorithm are known.  
In particular, on graphs with maximum vertex degree~$\Delta$, MIN achieves approximation ratio at least~$3/(\Delta+2)$, meaning for any graph an independent set of size this fraction or better times the optimal~\cite{halldorsson1994greed} is output.   
Interestingly, the theoretical analysis doesn't rely on choosing the minimum degree vertex to delete; instead the bounds just rely on the vertex being a critical (near minimal) vertex \footnote{A vertex $v$ is said to be critical if its graph degree is at most the average degrees its neighboring vertices.}.

We remark that the classical greedy algorithms we consider here are not necessarily state of the art for the general MIS problem.  While other better algorithms may exist in different problem settings, we compare to greedy algorithms because of their close similarity to Iterative Quantum Algorithms~\footnote{Indeed, the greedy algorithms considered can be viewed as special cases of IQA where the quantum device is ignored and a purely classical selection rule is employed.}.

\section{Iterative Quantum Algorithms for Max Independent Set}
\label{sec:general_theory}

To encode the MIS problem, let $x_i\in\{0,1\}$ indicate membership of vertex $i$ in the candidate solution set with $i=1,\dots, n=|V|$. Then the size of the set corresponding to the bit string~$x$ can be expressed as $r(x)=\sum_{i=1}^n x_i$. 

In addition to the reward function~$r(x)$, the independent set constraint requires that each edge has at most one of its vertices set to one.  This can be encoded as the kernel of the penalty function $p(x) = \sum_{\langle i, j\rangle} x_i x_j$, where the sum is over all pairs of vertices that share an edge, denoted by $\langle i,j\rangle$.  The constrained optimization problem can then be transformed into an unconstrained optimization problem, using a positive Lagrange multiplier $\lambda$ such that we seek to minimize
\begin{align}
    \label{eq:c_classical}
    c(x) &= -2\,r(x) + 4\,\lambda\, p(x)\nonumber\\
    &= -2\sum_{i=1}^n x_i + 4\,\lambda \,\sum_{\langle i, j\rangle} x_i x_j.
\end{align}
Here the constant factors are taken for convenience when we convert to qubits shortly.

The goal is then to minimize $-r(x)$, subject to the hard constraint that the state satisfies $p(x)=0$. We next convert this to a Hamiltonian~$\hat{C}$ acting on qubits which we can do by encoding the $x_i$ as Pauli-$z$ operators.  
Here we follow the convention $x_i \to \frac{1}{2}(I+\sigma_i^{(z)})$ which maps the $1$ value of the $i$th bit to the $+1$ eigenvalue of $\sigma^{(z)}$ for the $i$th qubit, and similarly $0$ to $-1$.

For quantum approaches, typically we seek to create a quantum state that minimizes the expectation value, $\langle \hat{C}\rangle$, which we then repeatedly prepare and measure in the computational basis in order to sample from possible solution candidates. 
We remark that in practice any infeasible bit strings obtained can always be corrected to yield an independent set of size at least that of the penalized cost value, see App.~\ref{app:MIScorr}.

The problem Hamiltonian for MIS then becomes
\begin{equation}
    \label{eq:C}
    \hat{C} = -\sum_{i=1}^n (\sigma_i^{(z)}+I)+\lambda \sum_{\langle i, j\rangle} (\sigma_i^{(z)}+I)(\sigma_j^{(z)}+I).
\end{equation}

\subsection{Iterative Quantum Optimization}

Iterative Quantum Algorithms are a new paradigm for boosting the performance of noisy, low-depth quantum algorithms using an interactive approach.  
While recent proposals~\cite{Bravyi2020,dupont2023quantum} are primarily based on QAOA, it is important to remark that similar iterative approaches are easily extendable to incorporate a wide variety of potential quantum ans\"atze, algorithms, and devices~\cite{IQO}. Nevertheless for our MIS application considered in the remainder of this paper we will focus on the QAOA case, with the exception of a brief foray into Iterative Quantum Annealing in Sec.~\ref{ssec:annealing}.

RQAOA~\cite{Bravyi2020} was based off of the standard QAOA ansatz which starts from the uniform superposition state 
\begin{equation}
    \label{eq:varphi}
    \ket{\varphi_0} = \frac{1}{\sqrt{2^n}}\sum_{z\in\{-1,1\}^n} \ket{z},
\end{equation}
where we label computational basis vectors $\ket{z}$ by their binary representation.  Given parameters $\vec{\gamma}$ and $\vec{\beta}$, the problem Hamiltonian~$\hat{C}$, and a simple mixer Hamiltonian~$\hat{B}$, the QAOA state is then 
\begin{equation}
    \ket{\psi(\vec{\beta},\vec{\gamma})} = e^{-i\beta_p\hat{B}}e^{-i\gamma_p \hat{C}}
    \dots
    e^{-i\beta_1\hat{B}}e^{-i\gamma_1 \hat{C}}
    \ket{\varphi_0}.
\end{equation}
The input $p$ determines the circuit depth, and the parameters $\gamma,\beta$ can be optimized variationally or by other means. 
RQAOA took QAOA applied to MaxCut problems as a component, and applied another classical outer loop on top of it.  This outer loop looked at the two-bit correlators (expectation values $\langle \sigma_i^{(z)}\sigma_j^{(z)}\rangle $) for all pairs of qubits sharing an edge in the graph.  It selects the edge with the largest correlator in absolute value and then reduces the size of the problem by fixing those qubits to be either correlated or anti-correlated (depending on the sign of the correlator). In particular, fixing $\sigma_i^{(z)}\sigma_j^{(z)}=\pm 1$ can be implemented by direct substitution of $\sigma_i^{(z)}=\pm \sigma_j^{(z)}$ into the cost Hamiltonian $\hat{C}$, which always reduces the problem by at least one variable. The process is then repeated on this increasingly smaller graph until the problem becomes trivial to solve.   
Then, the sequence of fixing rules is reversed to recover the candidate solution to the input problem.

While more general cost functions were briefly considered for RQAOA in \cite[App.~C]{Bravyi2020}, numerous aspects were not explored in detail, in particular how to apply the algorithm to different classes of constrained optimization problems. In \cite{IQO} we detail how in the constrained case, the simple variable substitution of RQAOA generalizes to \textit{logical inference}, which can lead to more substantial reductions depending on the reduction rules considered. Indeed, as we will see shortly for MIS, fixing a vertex to be in the candidate set leads to a more significant problem reduction than the case of removing a vertex from the graph. 
RQAOA and related iterative approaches also easily incorporate alternative selection and reduction rules, for instance we are by no means limited to two-bit correlators that correspond to terms of the cost Hamiltonian. 
For MIS, estimating and fixing vertex values is a natural reduction step in terms of easily maintaining the independent set constraint.
Hence, our hybrid quantum algorithms can be viewed as generalizations of classical greedy algorithms for MIS to single-bit expectation values. We further consider the inclusion of two-bit correlators for MIS in Sec.~\ref{sec:quadraticRules}.

\subsection{Quantum-Enhanced Greedy Algorithm for MIS} \label{sec:qalg}

Our iterative algorithm follows the same structure as the greedy algorithm given above, with only the first step modified to be:
\begin{enumerate} 
   \item[1'.] 
   Choose a vertex in the graph with the largest single variable expectation value
\begin{equation}
    J_j(\vec{\beta},\vec{\gamma}) = \bra{\psi(\vec{\beta},\vec{\gamma})}\sigma_j^{(z)}\ket{\psi(\vec{\beta},\vec{\gamma})}
\end{equation}
with respect to the QAOA state $\ket{\psi(\vec{\beta},\vec{\gamma})}$.
\end{enumerate}
Here at each iteration we assume the QAOA state has been updated for the reduced problem Hamiltonian and suitably optimized. The single-bit correlators can then be estimated from repeated preparation and measurement of the state.

We remark that, distinct from RQAOA, our selection procedure considers actual instead of absolute values.
Furthermore, 
we employ a \textit{one-sided} selection rule 
in that selected vertices are fixed to be in the set but generally never fixed to not be.

In the next subsection, we will analyze this $J_j$ for $p=1$ depth QAOA to determine how it behaves compared to the classical greedy algorithm. 
We find that Iterative QAOA-1 uses the same graph degree information as MIN such that we expect closely related performance between the two algorithms. We remark that both algorithms contain a hidden degree of randomness in how ties are broken, i.e. when there are multiple vertices of the same minimal degree or maximal $J_j$ value. We assume ties are either broken randomly, or with respect to a fixed ordering of the graph vertices. Clearly, this can lead to different performance across the same or permuted instances. To alleviate this effect we use the same fixed ordering to break ties across the different algorithms considered in our numerical results to follow.

\section{Analysis for $p=1$ QAOA}
\label{sec:p=1}

In this section we analyze the output of QAOA-1 for MIS problems.  We consider QAOA starting from the equal superposition state $\ket{\varphi_0}$ of Eq.~\ref{eq:varphi}, and utilizing the standard transverse-field mixer 
\begin{equation}
    \label{eq:B}
    \hat{B} = -\sum_{i=1}^n \sigma_i^{(x)},
\end{equation}
here negated for minimization.
Throughout, the notation $\sigma_i^{(s)}$ represents the Pauli matrix for the $i$th qubit along the $s$ axis, $s=x,y,z$.  

For our iterative algorithm with QAOA-1, 
the objects we care about will be the expectation values of each of the bits in the QAOA state, given by
\begin{equation}
    \label{eq:Ji_initial}
    J_j(\beta,\gamma) = \bra{\varphi_0}e^{i\gamma \hat{C}}e^{i\beta \hat{B}}\sigma_j^{(z)}e^{-i\beta \hat{B}}e^{-i\gamma \hat{C}}\ket{\varphi_0}.
\end{equation}
For QAOA-1, explicit formulas for observable expectation values can be obtained~\cite{wang2018quantum,hadfield2018quantum,ozaeta2022expectation, marwaha2022bounds}. 
In Appendix \ref{app:Zpathsum}, we use a path-sum approach to calculate an exact closed form for $J_j(\beta,\gamma)$ in terms of the graph degrees of the vertices, $d_j$:
\begin{align}
    \label{eq:Ji_final}
    J_j(\beta,\gamma) = \sin(2\beta) 
    (\cos(2\gamma\lambda))^{d_j}\sin(2\gamma(1-d_j\lambda)).
\end{align}
Independent of the values of $\gamma$ and $\beta$, there is one clear conclusion from this result: Iterative QAOA-1 in selecting the qubit with the highest $J_j$ will only be using information about the degree of the vertices.  Therefore, Iterative QAOA-1 will not utilize any information that was not already accessible to the classical greedy algorithm.  

Moreover, under reasonable assumptions on $\beta,\gamma$  the magnitude of $J_j(\beta,\gamma)$ is seen to decrease with increasing $d_i$, such that we expect the lowest degree vertices to be heavily favored. We remark that rigorously proving this would require characterizing the optimal angles for all possible problem instances which appears challenging. 
While it is of course possible that the QAOA-1 parameters are set such that the iterative outer loop picks a vertex other than the lowest degree vertex, our numerical simulations discussed below indicate that in practice, the lowest degree vertex is always picked. Hence the performance of Iterative QAOA-1 is seen to mimic that of the classical greedy algorithm MIN.

Of course, if we instead considered Iterative QAOA-$p$, then as we increased $p$ the vertex expectation values depend in increasingly complicated ways on the distance-$p$ neighborhood in the graph for each vertex~\cite{Farhi2014,wang2018quantum}, and the same approach does not appear tractable in general. 
We further discuss the $p>1$ case in Sec.~\ref{sec:higherdepth} and 
consider the possible QAOA depth required to achieve a quantum advantage. 
Nevertheless the $p=1$ case provides a valuable guide for helping to understand the behavior of iterative quantum algorithms~\footnote{While, as we have seen, the lowest depth quantum expectation values can be efficiently computed classically, efficient classical sampling from said circuits is not believed to be possible~\cite{Farhi2014,farhi2016quantum} which hints at further possibilities for quantum advantage in practice.}.

\subsection{Comparing to MAX}

Above we considered the greedy algorithm MIN that chooses the lowest degree remaining vertex to be part of the solution set.  An alternative greedy algorithm is MAX which iteratively chooses the highest degree vertex to be deleted from the graph until only disconnected vertices are left, constituting an independent set.

Our quantum algorithm above can be modified to mimic MAX instead at low depth through the simple alteration of ranking the graph vertices by lowest (generally most negative) $J_j$ value.  Then whichever vertex is lowest is removed from the graph as in MAX.  Since the $J_j$ value still only depends on the degrees of the vertices, this still has as much information as classical MAX.  

Due to their similarity to classical MIN and MAX, we refer to the Iterative Quantum Algorithm fixing bits to be one as ``MINQ,'' and to the the variant that fixes bits to be zero as ``MAXQ.''
In our numerics below, we consider both MAX and MIN variants for comparison.

\subsection{Numerical Verification} 
For numerical verification of these results, in particular that for $p=1$ the quantum and classical algorithms should perform the same, we swept through a number of qubits $n=5$ up to $n=18$. 
For each $n$ we generated 400 (often more in testing) unweighted, connected Erd\"os-R\'enyi graphs, chosen at random with the probability of including each edge being $q = 1.2\ln(n)/n$. (Note that repeats were allowed, and that for small $n$ this just included all possible graphs.)

For each instance we ran Iterative QAOA-1 as described above on a noiseless classical simulator, optimizing the $\gamma$ and $\beta$ angles using Nelder-Mead.  We ran both a MINQ version that froze the qubit with the highest Pauli-$z$ expectation value to be one (and all its neighbors to be $-1$) and a MAXQ version that froze the qubit with the lowest value to be $-1$.   For each reduction step, we checked whether the frozen qubit was indeed in the set of lowest or highest degree vertices, respectively for MINQ and MAXQ. If Iterative QAOA-1 ever made a choice that deviated from a choice that could have been made by the corresponding classical greedy algorithm, we designated a flag to be flipped to highlight that something not explained by the classical algorithm was happening.  Across all our trials with $p=1$, this flag was never flipped.

In Sec.~\ref{sec:weightedCase} we will describe the same algorithm generalized to weighted graphs.  In numerical experiments for that case, this flag was often flipped. 

We consider empirical algorithm performance for Iterative QAOA-2 in Sec.~\ref{sec:numericalPerf}.

\section{Higher Depth Iterative QAOA}
\label{sec:higherdepth}

For general depth $p$ QAOA, the required expectation values for our iterative approach are 
\begin{equation}
    J_j(\vec{\beta},\vec{\gamma}) = \bra{\varphi_0}\hat{U}^\dagger(\vec{\beta},\vec{\gamma})\sigma_j^{(z)}\hat{U}(\vec{\beta},\vec{\gamma})\ket{\varphi_0},
\end{equation}
where $\hat{U}(\vec{\beta},\vec{\gamma}) = \prod_{i=1}^p e^{-i\beta_i\hat{B}}e^{-i\gamma_i \hat{C}}$. 
For this general setting obtaining closed form results analogous to the $p=1$ case is challenging. A path-sum or other approach can again be employed, but quickly becomes unwieldy to the point where performing similar analysis appears impractical. 
Though slightly more manageable expressions can be obtained in restricted settings, such as  triangle-free and bounded degree graphs of small degree, similar difficulties persist~\cite{hadfield2018quantum}.

\subsection{Iterative QAOA-2}
Though obtainable, even expressing the $p=2$ version of Eq.~(\ref{eq:Ji_final}) is too long to fit reasonably in this paper \footnote{For this general calculation we used symbolic manipulation in Mathematica.}.  
The end result is the $J_i$ values now depend on properties of all vertices in the distance $1$ neighborhood of the $i$th vertex, in particular but not limited to the degrees of the neighboring vertices and number of triangles in the graph containing vertex $i$.

In other words, Iterative QAOA-2 takes into account information about the local neighborhood around each vertex when ranking them. 
All of this information is accessible to classical algorithms, and there are  algorithms that take into account a portion of this information when making a decision.  However, by the time we reach $p=2$, the quantum algorithm is taking into account this information in a non-trivial way that is not easily translatable to a classical equivalent, and as the list of information considered grows, it quickly becomes unclear what appropriate classical heuristics would even be.

In App.~\ref{sec:qaoa2} we apply an alternative analysis approach~\cite{hadfield2018quantum} based on conjugation by Pauli operator to illuminate this. Though the aforementioned difficulties remain, we apply small-angle analysis~\cite{hadfield2022analytical} to show that for QAOA with $p=2$, the leading-order expression for $J_j(\beta_1,\beta_2,\gamma_1,\gamma_2)$ 
with respect to $\gamma_2$ 
is given by
\begin{align}
&\cos(2\beta_2)\sin(2\beta_1) \sin(2\gamma_1 (1-\lambda d_j)) \cos^{d_j}(2\gamma_1 \lambda)\\
&-\nonumber \sin(2\beta_2)\cos(2\gamma_2 c_j) \cos^{d_j} (2\gamma_2 \lambda)\,\times
 \\
&\,\,\,\,\bigg( \cos(2\beta_1) \sin(2\gamma_1 (\lambda d_j -1)) \cos^{d_j}(2\gamma_1 \lambda) \nonumber\\
&\,\,\,\,- \nonumber (2\gamma_2 \lambda) \cos(2\gamma_1 (\lambda d_j-1)) \cos^{d_j}(2\lambda\gamma_1) \,\, + \dots \bigg).
\end{align}
Hence to lowest order we see that the value of $J_j(\beta_1,\beta_2,\gamma_1,\gamma_2)$ for $p=2$ is determined by a linear combination of its $p=1$ value $J_j(\beta_1,\gamma_1)$, and a correction factor. 
Observe that these terms again depend only on the degree of vertex~$j$. As $\gamma_2$ is increased, additional significant terms appear in the correction factor. Term at the next lowest order depend on the number of triangles in the graph containing vertex $j$. Subsubsequent terms depend on vertex subsets containing $j$ and $3$ of its neighbors, and so, which leads to combinatorial challenges in analyzing the general case.

\subsection{Numerical Performance} 
\label{sec:numericalPerf}

In  Fig.~\ref{fig:p=2_unweighted} we plot the approximation ratio (achieved set size divided by true maximum independent set size) averaged over two hundred random, connected Erd\"os-R\'enyi graphs with edge probability $q = 1.2\ln(n)/n$, with the error bars depicting the standard error of the mean.  The true maximum independent set size was determined through a brute force exhaustive search that was feasible at these small problem sizes.  All the Iterative Quantum Algorithms in this plot are run with $p=2$ QAOA as the quantum component with a constraint weighting $\lambda = 1$.  Note that QAOA itself
does not appear on this plot because the expected QAOA approximation ratio ranges from $0.69$ at $n=5$ down to $0.495$ at $n=16$.

The quantum algorithms are labeled here with a ``Q,'' with MINQ and MAXQ referring to the iterative quantum algorithms with the same bit fixing rules as their classical counterparts.  MMQ is an iterative quantum algorithm (see Section \ref{ssec:onoff}) that incorporates elements of both MIN and MAX. 

From our results in the figure it is obvious that the Erd\"os-R\'enyi graphs we are using are better served by MIN, as expected, though this situation could change for other models or problem variants. 
For QAOA with $p=1$ we explained and verified the same performance between MINQ and MIN.  
For $p=2$ we see that the iterative quantum algorithms consistently outperform the corresponding classical algorithm, to a significant degree, especially MINQ which obtains approximation ratio very close to $1$. 
While this is an exciting indicator, unfortunately, the system sizes accessible via our simulations are small enough that all algorithms (except for QAOA itself) are performing very well, and the sizes are much too small to attempt extrapolation in a meaningful way; we expect orders of magnitude increase in problem size may be required to access regimes truly challenging for classical algorithms. 
Nevertheless these results point to the promise of iterative quantum optimization on increasingly larger quantum devices.

Numerical simulations of iterative quantum algorithms for unweighted MIS with $p=3,4,5$ 
QAOA do not show any statistically significant departure from the behavior of algorithms with $p=2$ QAOA for the range of $n$ we consider.  This story does change however for Weighted MIS as we discuss in Section \ref{sec:modifications}.

\subsection{Depth Scaling with Width of Graph}

While we do not explicitly work out the path-sum formulation for QAOA with $p>2$, the results of these calculations and the general structure of how the path-sum works make it obvious how this behavior will scale.  A depth $p$ QAOA algorithm when ranking vertices to eliminate in the iterative step will consider all information about a selected vertex, all vertices connected to it by paths of length $p-1$, and the full connectivity of those vertices to each other and other parts of the graph.

Notably, this implies that when $p$ is comparable to the graph diameter, iterative QAOA will be ranking each vertex using information from the entire graph.  Given that the width of a random $n$ vertex graph typically grows as $\log{n}$, this gives hope that QAOA or its iterative versions might be able to tackle Maximum Independent Set with depth only logarithmic in the size of the graph.  This hope should of course be tempered~\cite{Bravyi2020,farhi2020quantum_t,farhi2020quantum_w} because while the algorithm is able to consider the entire structure of the graph, it still only has logarithmically many variational parameters with which to work.  Just because all graph properties are being included in the calculations does not mean that $\mathcal{O}(\log{n})$ parameters allow for enough expressivity to use that information and rank the vertices appropriately.
At the same time, restricting to relatively few parameters is desirable for the practical scalability of variational algorithms such as QAOA with respect to parameter optimization challenges.

We remark that for fixed QAOA depth $p$, as the problem becomes subsequently smaller in size, and hence easier, we expect QAOA performance to 
relatively improve. Furthermore, smaller problems eventually facilitate deeper circuits using the same quantum resources that can yield further improvements.

While these observations regarding the scaling or vertex ranking for iterative QAOA are encouraging for the long term prospects of IQAs, this should not be taken as conclusive evidence of their superiority.  These algorithms have in principle the ability to consider more information than a classical greedy algorithm could.  It encouragingly should be noted that in order to even enumerate all these features in the case above a classical algorithm would require time scaling at least proportional to $n$, much longer than the $\mathcal{O}(\log{n})$ layers of the quantum algorithm.
Of course mere consideration of all these features does not necessarily imply that the features are being used efficiently to create the ranking.  However, as discussed, analytical means of discovering how the quantum algorithm is exactly using all these features for $p$ in this regime appears impractical.

\section{IQO Generalizations}
\label{sec:modifications}

Here we consider further variants to our Iterative Quantum Algorithm aiming toward obtaining even better performance. 
The most obvious extension to go beyond the classical greedy algorithm is of course to consider increasing QAOA depths to $p\gg 1$.  As we discussed in the previous section, this serves to extend the local search space that the algorithm considers around each vertex, until the algorithm is eventually able to consider the entire graph.   
While deep circuits are a path to go beyond the power of classical algorithms generically, assuming adequate quantum hardware has become available, it comes with additional challenges such as those related to parameter setting~\cite{Cerezo2021,bharti2022noisy}. Especially in cases where only a local optima in parameter space can be reasonably obtained, iterative approaches offer a means of potentially improving upon sampling of the state directly.

There are numerous variants of IQO that can be employed toward more sophisticated algorithms. 
These variants do not have analytic guarantees that they perform better than MIN or $p=1$ MINQ, only that they consider information in a way that is not efficient to mimic on classical computers.

\subsection{Fixing Bits Both On and Off}
\label{ssec:onoff}

In Section \ref{sec:p=1} we described methods that mimic classical greedy MIN and MAX algorithms via iterative quantum algorithms that look for the bits with the highest or lowest single-variable 
expectation values, respectively, and fix that variable accordingly.  
We could also combine these two rules such that at any given step, we can choose to fix either a high bit to be $1$ (and then fix all of its neighbors to be $0$), or fix a low bit to be $0$, depending on the values of the single-bit expectation values.

In the first case where we fix a high bit $x_i=1$, the problem is reduced by $d_i +1$ vertices, and all edges involving these vertices are removed from the problem. 
This reduces the problem Hamiltonian much more significantly than if we had just substituted $\sigma^{(z)}_i\to\pm 1$ directly.
In the second case, a single vertex and its edges are removed. 
Both rules are logical implications of the independent set constraint.

The combined approach of choosing whether to fix high or low bits is also possible in a classical greedy algorithm, but it is  not immediately obvious how the algorithm should choose between a high or low degree vertex at a given step.
On the other hand, as we have shown, for $p=1$ QAOA the lowest degree vertices in the graph will correspond to the highest $J_j$ and the highest degree vertices to the lowest $J_j$.  To determine which to choose we could choose a vertex to add to our independent set (the same way as MIN) versus one to exclude from our independent set (the same way as MAX) based off the sign of the $J_j$ with the highest absolute value.  
The relative magnitudes of these quantities relate to how certain the quantum algorithm is, and even for depth $p=1$ QAOA this information isn't derivable from the degrees of the nodes alone but also depends on the global problem through the optimization of the $\gamma$ and $\beta$ angles.

To design a classical greedy algorithm that mimics this behavior, one possible choice would just be to use Eq.~(\ref{eq:Ji_final}) directly in the classical algorithm, where the greedy choice of high or low is with respect to the 
absolute value of $J_j(\beta,\gamma)$.
The use of Eq.~(\ref{eq:Ji_final}) would necessitate a classical optimization of the $\gamma$ and $\beta$ angles to minimize $\avg{\hat{C}}$ which could be done by employing the closed forms in Eq.~(\ref{eq:Ji_final}) itself and the two qubit correlators derived in App.~\ref{app:2bit}.  
Clearly, alternative discriminators could also be considered. 
We remark that iterative algorithms based instead on classical probabilistic sampling will naturally have metrics 
analogous to $J_j$.

In various figures throughout this paper, we include this expanded iterative quantum algorithm for comparison to other methods.  We refer to this method as ``MINMAXQ'' or ``MMQ'' for short since it incorporates elements of both classical MIN and MAX. While less effective than MIN for Maximum Independent Set, the MMQ generalization is important for the Weighted MIS problem variant we consider next.

\subsection{Weighted Case} \label{sec:weightedCase}
Consider Weighted MIS, the generalized problem where each vertex additionally comes with a reward weight $r_i$ for its inclusion in the independent set.
In this case the problem Hamiltonian is given by  
\begin{equation}
    \hat{C} = -\sum_{i=1}^n r_i (\sigma_i^{(z)}+I)+\lambda\sum_{\langle i, j\rangle} (\sigma_i^{(z)}+I)(\sigma_j^{(z)}+I).
\end{equation}
Proceeding as for the unweighted case, the single-variable expectation values used in the ranking now become  
\begin{align}
    \label{eq:Ji_weighted}
    J_j(\beta,\gamma) = \sin(2\beta) (\cos(2\gamma\lambda))^{d_j}\sin(2\gamma(r_j-d_j\lambda)).
\end{align}
This is a relatively minor change compared to Eq.~\ref{eq:Ji_final}, but the formula now provides a unique combination of how $d_j$ and $r_j$ are being handled that is distinct from how simple greedy algorithms behave.
Classically, both MIN and MAX can be expanded to weighted versions, (W)MIN and (W)MAX respectively \cite{Sakai2003}.  For weighted MIN, rather than picking the vertex, $v_i$, with the lowest degree $d_i$ to add to the proposed set, we instead pick the vertex with the highest ratio $r_i/(d_i+1)$.  Similarly for weighted MAX, we now eliminate the vertex with the lowest ratio $r_i/(d_i(d_i+1))$ from the graph.

\begin{figure}
\includegraphics[width = 0.48\textwidth]{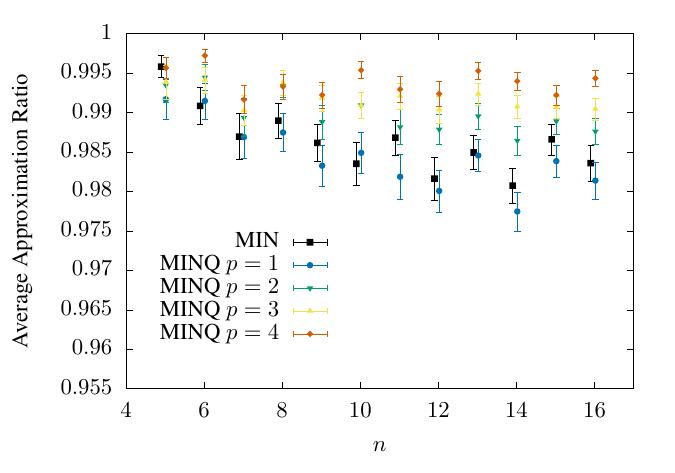}
\caption{The approximation ratio for various models as a function of $n$ for Erd\"os-R\'enyi graphs with $q = 1.2\ln(n)/n$ and vertex weights chosen uniformly at random in the range $[1,2]$.  Each data point is averaged over two hundred randomly chosen graphs, and the error bars show the standard error of the mean.  The MINQ algorithms are quantum iterative algorithms with QAOA at the specified $p$ depth as the quantum component.  
}
\label{fig:sweep_p_weighted}
\end{figure}

In Fig.~\ref{fig:sweep_p_weighted} we plot the performance of (W)MINQ 
for weighted MIS problems where each vertex is assigned a weight uniformly at random from the continuous range $[1,2]$, for various QAOA depths to show how increasing $p$ effects the algorithm.  Note that in contrast to the unweighted case, here MIN and MINQ with $p=1$ are distinct algorithms, though empirically there does not appear to be much significant difference between them.  For $p>1$, there is a statistically significant difference from MIN, but the returns may be diminishing for moderate increases in $p$. 
Also note that while our results indicate that higher $p$ iterative quantum algorithms take more information into account, this does not mean that their performance will be strictly better than classical greedy algorithms (or even than lower depth MINQ algorithms).  It is intuitive to expect them to perform better, but there can be cases, such as with $n=5$ here, where that is not the case.

\begin{figure}
\includegraphics[width = 0.48\textwidth]{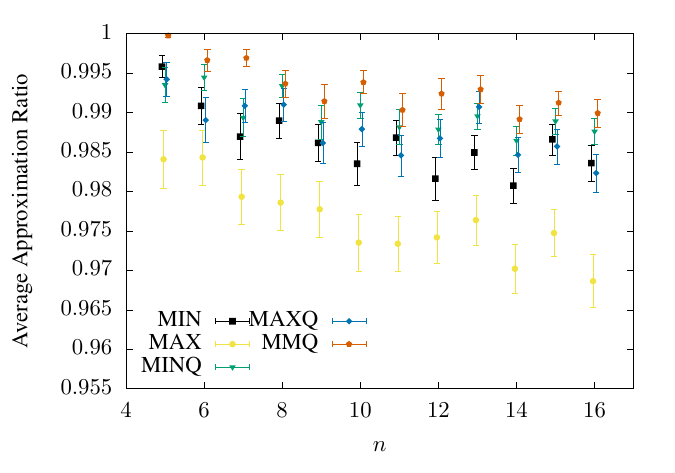}
\caption{
The approximation ratio (achieved set weight divided by the true maximum independent set weight) for various models as a function of $n$ for Erd\"os-R\'enyi graphs with $q = 1.2\ln(n)/n$ and vertex weights chosen uniformly at random in the range $[1,2]$.  Each data point is averaged over two hundred randomly chosen graphs, and the error bars show the standard error of the mean.  The ``Q'' algorithms are all performed with $p=2$ QAOA as the quantum component, and ``MinQ'' and ``MaxQ'' refer to the quantum algorithms that reduce to the corresponding classical greedy algorithm at $p=1$.  ``MMQ'' is a iterative quantum algorithm that chooses whether to fix bits on or off at each step as described in Section~\ref{ssec:onoff}.  The offsets along the $x$-axis are for visual clarity; all $n$ values are integers.}
\label{fig:p=2_weighted}
\end{figure}

In  Fig.~\ref{fig:p=2_weighted} we mirror the numerical results of Fig.~\ref{fig:p=2_unweighted}  but for the same class of weighted MIS problems as in Fig.~\ref{fig:sweep_p_weighted}. 
Interestingly, in the weighted case, $p=2$ MINQ is no longer performing significantly better than MIN.  What is obvious though is that MMQ is performing statistically better than any other algorithm tested, emphasizing the point that the preferred algorithm should depend on the problem details.  Since MMQ with $p\geq2$ is again performing selections that are not easily replicable in a classical setting, this provides evidence toward some quantum utility; although, these system sizes are still much too small for any definitive statements.

\subsection{Modifying the Selection Rule}
Another important variant of Iterative Quantum Algorithms is to consider more general selection rules. 
QAOA or another quantum algorithm can be used to rank vertices, edges, or hyperedges in any way desired, and the ranking can be used to reduce the problem. 
In addition to purely looking at the highest or lowest $J_i(\beta,\gamma)$ in the case of Weighted MIS, it is also quite natural to pick the highest or lowest weighted expectation value, $r_i J_i(\beta,\gamma)$, which more closely correspond to the change in objective function value resulting from each possible choice. More generally, other quantum observables can be considered, for instance both single-bit expectations and two-bit correlators could be included, with selection based on some function of their expectation values. 
These modifications lead to different ranking rules but do not otherwise modify the basic algorithm. 
We emphasize that the observables considered are not limited to terms that appear in the problem Hamiltonian, and that while the reduction step often matches the selection rule, this is not strictly necessary. 
The suitability of such selection rule modifications is problem dependent and we save a more thorough investigation for future work~\cite{IQO}.

We remark that for the original RQAOA prescription as applied to the MaxCut, for most input instances the problem is quickly transformed into a sequence of Weighted MaxCut instances. Hence the weighted selection rule can also be applied there. A distinct example of a selection rule for this problem is employed in the algorithm of Ref.~\cite{dupont2023quantum}.

\subsection{Quadratic Fixing Rules} 
\label{sec:quadraticRules}

The problem reduction step can also be modified, with a variety of possible choices, often designed in tandem with the selection rule. As mentioned previously, these include one-sided variants, where a variable or correlation can be fixed to one value but not the other. For constrained problems, clearly, reduction rules must take into account the problem constraints, and can often be designed in a problem structure-preserving way. 
Here we present a modification to our algorithm that goes beyond the single-bit setup considered above.
 
Specifically, we propose a one-sided iterative quantum algorithm that resembles RQAOA~\cite{Bravyi2020} by considering expectation values corresponding to edges in the graph instead of vertices:
\begin{enumerate}
    \item Run QAOA on the problem and measure the resulting two bit correlators $J_{ij}(\beta,\gamma) = \bra{\psi(\beta,\gamma)}\sigma^{(z)}_i\sigma^{(z)}_j\ket{\psi(\beta,\gamma)}$ for all pairs of vertices $i$ and $j$ that share an edge.
    \item Find the edge $(k,m)$ with the lowest (most negative) edge ZZ correlator, $J_{k\ell}$.
    \item Fix the two vertices to be anti-correlated, so that $z_\ell = -z_k$ and make this replacement in the Hamiltonian.
    \item Repeat steps (1)-(3) until the problem is trivial to solve for the remaining vertices.
    \item Back-propagate the information from step (4) through the anti-correlation decisions made at the occurrences of step (3).
\end{enumerate}  
Here, the reduction of Step 3 implies that exactly one of the vertices $k,\ell$ is in the solution set. Note this is no longer the case if we relax the restriction to $J_{ij}$ values corresponding to edges in the graph. 

An important observation with this quadratic algorithm variant is that, as opposed to the single-variable algorithm of Sec.~\ref{sec:qalg}, the output is no longer guaranteed to produce an independent set.  Nothing in this selection rule and reduction via simple substitution necessitates independence of the set.  Nevertheless we can always correct such a bit string to obtain some independent set of comparable cost; see the discussion in App.~\ref{app:MIScorr}.

This algorithm can be further modified analogously to the variants for the single-variable selection case proposed above. 
For instance, we could consider looking for the highest $J_{km}$ which implies correlation.  Since the only way two adjacent vertices can be correlated in a MIS problem is to have them both be logically zero (not in the independent set), this would force us to set both $x_k=x_\ell=0$.  We would then continue until we have a disconnected graph.  
With this variant we are guaranteed to output an independent set, but this algorithm is fundamentally just an amped up version of MAXQ described above where two variables are fixed per iteration instead of one.

In Appendix~\ref{app:2bit} we show the derivation of $J_{ij}(\beta,\gamma)$ for $p=1$ QAOA using the path-sum method.  The resulting correlators depend just on the degrees of the two vertices in the edge as well as the number of triangles the edge is involved in.  This is again all easy information to calculate and suggests the path toward a classical algorithm that could mimic this quantum algorithm. For $p>1$ similar issues concerning analysis and simulation arise as discussed above. 

While we don't anticipate the quadratic variants discussed here to outperform MINQ, they illuminate a number of design choices to consider when constructing Iterative Quantum Algorithms, and similar ideas should be relevant for other problems. 

\subsection{Iterative Quantum Annealing}
\label{ssec:annealing}

\begin{figure}
\includegraphics[width = 0.48\textwidth]{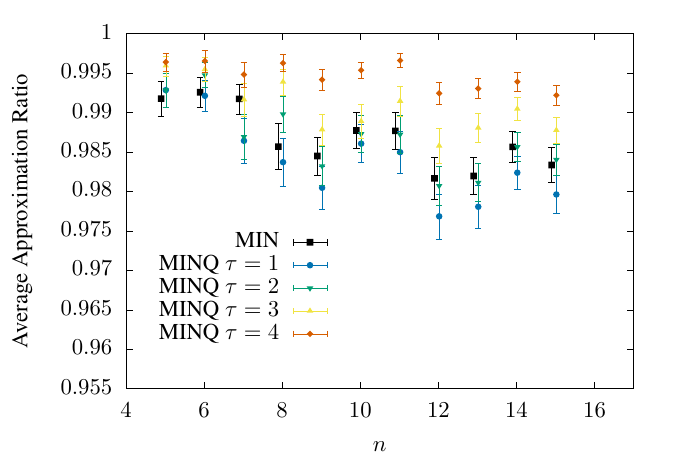}
\caption{The approximation ratio for various models as a function of $n$ for Erd\"os-R\'enyi graphs with $q = 1.2\ln(n)/n$ and vertex weights chosen uniformly at random in the range $[1,2]$.  Each data point is averaged over two hundred randomly chosen graphs, and the error bars show the standard error of the mean.  The MINQ algorithms are quantum iterative algorithms with Quantum Annealing at the specified runtime, $\tau$, the quantum component.  
}
\label{fig:sweep_p_qa_weighted}
\end{figure}

\begin{figure}
\includegraphics[width = 0.48\textwidth]{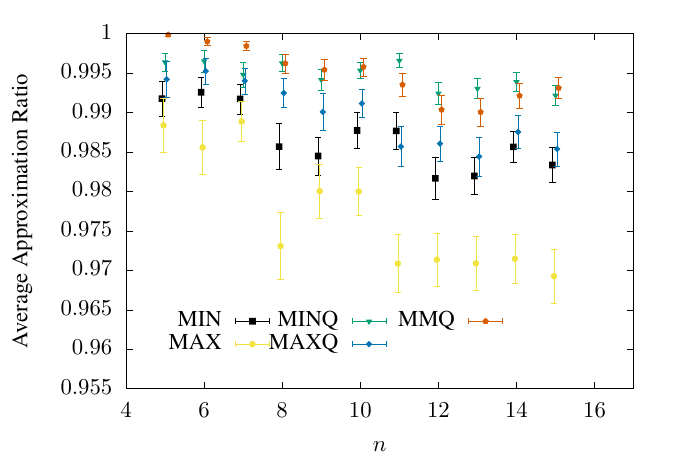}
\caption{
The approximation ratio (achieved set weight divided by the true maximum independent set weight) for various models as a function of $n$ for Erd\"os-R\'enyi graphs with $q = 1.2\ln(n)/n$ and vertex weights chosen uniformly at random in the range $[1,2]$.  Each data point is averaged over two hundred randomly chosen graphs, and the error bars show the standard error of the mean.  The ``Q'' algorithms are all performed with $\tau=4$ Quantum Annealing as the quantum component, and ``MinQ'' and ``MaxQ'' refer to the quantum algorithms that reduce to the corresponding classical greedy algorithm at $p=1$.  ``MMQ'' is a iterative quantum algorithm that chooses whether to fix bits on or off at each step as described in Section~\ref{ssec:onoff}.  The offsets along the $x$-axis are for visual clarity; all $n$ values are integers.}
\label{fig:tau=4_weighted}
\end{figure}

So far, we have considered QAOA as the quantum subroutine in iterative quantum optimization, but, as explained, any number of other quantum or classical choices are compatible. As an example, here we consider switching out QAOA for Quantum Annealing (a.k.a. Adiabatic Quantum Optimization)~\cite{Farhi2000,Kadowaki1998,albash2018adiabatic}, another family of quantum optimization algorithms that use roughly the same ingredients as QAOA, but are based on a continuous-time computational model instead of discrete quantum gates, and without a variational loop for parameter optimization. We remark that the idea of fixing problem variables using data output from a quantum annealing device was previously considered in~\cite{ayanzadeh2022quantum} in the context of finding the true ground state.

Quantum Annealing can be applied using the same problem Hamiltonian of Eq.~(\ref{eq:C}), driver (mixer) Hamiltonian of Eq.~(\ref{eq:B}), and initial state of Eq.~(\ref{eq:varphi}), with the goal of again minimizing $\avg{\hat{C}}$ for the final state. 
Instead of a bang-bang evolution, quantum annealing relies on an annealing schedule, $u(t)\in[0,1]$, such that the evolution can be described by a time-dependent Hamiltonian,
\begin{equation}
    \hat{H}(t) = u(t)\hat{C}+(1-u(t))\hat{B}.
\end{equation}
We will further use a linear annealing schedule, $u(t) = 1-t/\tau$, for a procedure running for total time $\tau$.  Appealingly, this algorithm has theoretical guarantees of finding the ground state for sufficiently large $\tau$, stemming from the quantum adiabatic theorem.

We use this instantiation of Quantum Annealing as a one-for-one replacement for QAOA in our Iterative Quantum Algorithms considered above.  The simulated results are shown in Figs.~\ref{fig:sweep_p_qa_weighted}~\&~\ref{fig:tau=4_weighted}, which exactly mirror Figs.~\ref{fig:sweep_p_weighted}~\&~\ref{fig:p=2_weighted} but with Quantum Annealing used to estimate the required expectation values instead of QAOA. 
(Note that these were done with newly generated problem instances, so the MIN and MAX performance also vary slightly.) 
Theoretical analysis of IQA is challenging with Quantum Annealing due to its continuous nature, but these plots indicate that in practice this switching of subroutines produces similar results.  Observe that $\tau\approx 2$ produces results similar to $p=1$, and $\tau\approx 4$ produces results similar to $p=3$.  Note that in all these cases, the base improvement of the iterative process is observed to be roughly equivalent between QAOA and Quantum Annealing, with similar performance in the base algorithm corresponding to similar performance in the iteratively enhanced algorithm. An important future research direction is to better understand generally the relationship between performance of the base algorithm and its iteratively enhanced variants.

\section{Conclusion}

We have developed novel Iterative Quantum Optimization algorithms for Maximum Independent Set. Our algorithmic framework generalizes RQAOA as well as other iterative quantum approaches. For the lowest QAOA depth $p=1$ realization of our algorithms, we show that our quantum algorithms perform identically to classical greedy algorithms, making the same choices as those algorithms at each iteration. For larger depths we expect still improved performance as QAOA improves with~$p$.  

While beyond the lowest depth realizations our hybrid quantum algorithms are seen to outperform the basic classical greedy algorithms, this shouldn't immediately be taken as evidence of quantum advantage.  As circuit depth increases, the quantum algorithms are accessing information about larger local neighborhoods around each vertex in the graph.  For constant quantum circuit depth, this information is typically readily accessible to classical algorithms.  The classical greedy algorithms we presented do not take the features of these larger neighborhoods into account, but other classical algorithms may. 
  
The benefit of the quantum algorithm is that it can naturally incorporate this extra information. As mentioned, global parameter optimization further incorporates problem structure in an implicit way. 
Additionally, we  
discussed hope that the 
Iterative QAOA circuits could access information about the full graph structure with circuit depth only $\mathcal{O}(\log n)$, 
though it has yet to be shown definitely whether merely considering these additional graph properties suffices for quantum advantage. We reiterate that our framework is compatible with a wide variety of quantum algorithms, for example generalizations of QAOA~\cite{hadfield2019quantum} that enforce problem hard constraints without the need for penalty terms, or hardware-efficient ans\"atze~\cite{leone2022practical,maciejewski2023design}, or even alternative paradigms such as quantum annealing~\cite{djidjev2018efficient}.

Our empirical results are intriguing, but ultimately, the problem sizes accessible with our simulations are too small to truly capture the scaling behavior of all of these algorithms.  It is clear that for low depths the iterative quantum algorithms greatly improve on the results of just the basic quantum algorithm that is embedded in their loop, but much of this power comes from the classical aspects of the hybrid iterative structure.

Iterative Quantum Algorithms provide a useful and potentially powerful way to extend current noisy intermediate-scale quantum algorithms, building off of underlying iterative algorithms and heuristics that also prove effective in purely classical settings. 
While the ability to potentially retain and improve upon classical performance guarantees is alluring, more work is needed toward better understanding the requirements for achieving bonafide quantum advantage. Our work takes important steps towards extending iterative quantum algorithms to wider classes of optimization problems. We remark that our iterative quantum approach for MIS also immediately yields approximation algorithms for the Maximum Clique and Minimum Vertex Cover problems via simple problem transformations, see e.g. \cite[App.~A]{hadfield2019quantum}.
An important future direction is to explore further classes of problems in detail. Finally, as quantum hardware continues to advance, the deployment of more sophisticated quantum circuits should empirically inform the design of even more effective iterative quantum algorithms and heuristics.

\acknowledgements

We thank our colleagues from NASA QuAIL for numerous helpful discussions, in particular M.~Sohaib Alam, Zhihui Wang, Zoe Gonzalez Izquierdo, Shon Grabbe, and Eleanor Rieffel.
This work was supported by the NASA AIST program and by the Defense Advanced Research Projects Agency (DARPA) under Agreement No. HR00112090058 and DARPA-NASA IAA 8839, Annex 114.
L.~T.~B.~is a KBR employee working under the Prime Contract No.~80ARC020D0010 with the NASA Ames Research Center.
S.~H. was supported under NASA Academic Mission Services under contract No. NNA16BD14C.

The United States Government retains and the publisher, by accepting the article for publication, acknowledges that the United States Government retains a non-exclusive, paid-up, irrevocable, worldwide license to reproduce, prepare derivative works, distribute copies to the public, and perform publicly and display publicly, or allow others to do so, for United States Government purposes. All other rights are reserved by the copyright owner.

\bibliography{MIS}

\begin{appendix}

\section{Expectation Value Formulas}
\label{app:Zpathsum}
Here we calculate a closed form for the single-variable Pauli-z expectation values of QAOA-1 on an MIS problem, starting from Eq.~(\ref{eq:Ji_initial}) and ending at Eq.~(\ref{eq:Ji_final}).

Let's start be explicitly writing out the initial states in the $\sigma^{(z)}$ eigenbasis and using these basis states to evaluate the $\hat{C}$ exponentials:
\begin{align}
    J_j(\beta,\gamma) = \frac{1}{2^n}&\sum_{z,z'\in\{-1,1\}^n} e^{i\gamma (C(z')-C(z))}\\\nonumber
    &\times\bra{z'}e^{i\beta \hat{B}}\sigma_j^{(z)}e^{-i\beta \hat{B}}\ket{z}.
\end{align}

Now we are going to focus on just the remaining quantum expectation value here.  One insertion of a resolution of the identity in the $\sigma^{(x)}$ basis should be enough here:
\begin{align}
    &\bra{z'}e^{i\beta \hat{B}}\sigma_j^{(z)}e^{-i\beta \hat{B}}\ket{z}\\\nonumber
         &=\sum_{x\in\{-1,1\}^n}\bra{z'}e^{i\beta \hat{B}}\sigma_j^{(z)}\ket{x}\bra{x}e^{-i\beta \hat{B}}\ket{z}.
\end{align}
The resolution of the identity can easily act on the rightmost exponential, and the $\sigma_j^{(z)}$ will act on the basis state to transform it into $x'$ which is the same as $x$, except that the $j$th bit is flipped.  Then this expectation value reduces to 
\begin{equation}
    \bra{z'}e^{i\beta \hat{B}}\sigma_j^{(z)}e^{-i\beta \hat{B}}\ket{z}=
         \sum_{x\in\{-1,1\}^n}e^{i\beta (B(x')-B(x))}\braket{z'}{x'}\braket{x}{z}.
\end{equation}
Since $x$ and $x'$ only differ in the $j$th bit, most of these elements in the sum won't matter.  In fact, for every bit except the $j$th bit, this just reduces to a delta function between $z'$ bits and $z$ bits:
\begin{align}
    \bra{z'}e^{i\beta \hat{B}}\sigma_j^{(z)}e^{-i\beta \hat{B}}\ket{z}=
    \\\nonumber
         \left(\prod_{i\neq j} \delta{z_i,z'_i}\right)\sum_{x_j = \pm1}e^{2i\beta x_j}\braket{z_j'}{-x_j}\braket{x_j}{z_j}.
\end{align}
Next remember that $\braket{x_j}{z_j} = \frac{1}{\sqrt{2}}(-1)^{(\frac{1+x_j}{2}) (\frac{1+z_j}{2})}$.  We can work this all out to get
\begin{align}
    \bra{z'}e^{i\beta \hat{B}}\sigma_j^{(z)}e^{-i\beta \hat{B}}\ket{z}=
    \\\nonumber
         \left(\prod_{i\neq j} \delta{z_i,z'_i}\right)\frac{1}{2}\left(e^{2i\beta}(-1)^{\frac{1+z_j}{2}}+e^{-2i\beta}(-1)^{\frac{1+z'_j}{2}}\right).
\end{align}

Next our goal will be to use all of this information in the expectation value formulations above
\begin{align}
    J_j(\beta,\gamma) &= \frac{1}{2^{n+1}}\sum_{z\in\{-1,1\}^n} 
    \sum_{z'_j=\pm z_j}
    e^{i\gamma (C(z')-C(z))}\\\nonumber
    &
    \times\left(e^{2i\beta}(-1)^{\frac{1+z_j}{2}} + 
            e^{-2i\beta}(-1)^{\frac{1+z'_j}{2}}\right).
\end{align}

The easiest thing here would be to do the summation over $z'_j$ directly
\begin{align}
    J_j(\beta,\gamma) &= \frac{1}{2^{n}}\sum_{z\in\{-1,1\}^n} \bigg(
    \cos(2\beta)(-1)^{\frac{1+z_j}{2}}
    \\\nonumber&
    +i e^{2i\gamma z_j\left(1-\lambda\sum_{i\in N(j)}(1+z_i)\right)}
    \sin(2\beta)(-1)^{\frac{1+z_j}{2}}
    \bigg).
\end{align}
We can now get rid of all summations over bits not involved in this:
\begin{align}
    J_j(\beta,\gamma) = &\frac{2^{n-d_j-1}}{2^{n}}\sum_{\substack{z_k=\pm 1\\k\in\{j\}\bigcup N(j)}} \bigg(
    \cos(2\beta)(-1)^{\frac{1+z_j}{2}}
    \\\nonumber&
    +i e^{2i\gamma z_j\left(1-\lambda\sum_{i\in N(j)}(1+z_i)\right)}
    \sin(2\beta)(-1)^{\frac{1+z_j}{2}}
    \bigg),
\end{align}
where $N(j)$ refers to the set of all vertices that have an edge connected to $j$.

It is easy to see that the $\cos$ terms all cancel out to zero nicely.  Similarly the sum over $z_k$ in the second portion simplifies a lot too:
\begin{align}
    J_j(\beta,\gamma) = &\frac{1}{2^{d_j}}\sum_{k\in\{j\}\bigcup N(j)} \bigg(
    \\\nonumber&
    \sin\left(2\gamma\left(1-\lambda\sum_{i\in N(j)}(1+z_i)\right)\right)
    \sin(2\beta)
    \bigg).
\end{align}

Very little remains here of the structure of the original problem, and frankly what remains can be worked out mostly with combinatorics.  We just need to count how many of the vertices in the local neighborhood of $j$ are $\pm1$.  We do this below
\begin{align}
    J_j(\beta,\gamma) = &\frac{1}{2^{d_j}}\sin(2\beta)\sum_{m = 0}^{d_j} \bigg(
    \\\nonumber&
    {d_j \choose m}\sin\left(2\gamma\left(1-2\lambda m\right)\right)
    \bigg).
\end{align}
This sum has a closed form:
\begin{align}
    J_j(\beta,\gamma) = \sin(2\beta)
    (\cos(2\gamma\lambda))^{d_j}\sin(2\gamma(1-d_j\lambda)).
\end{align}

\section{Leading-order terms of $J_j$ for $p=2$} 
\label{sec:qaoa2}
We next apply the Pauli Solver algorithm of \cite{hadfield2018quantum,wang2018quantum} as a complementary approach to the path-sum analysis.  We will revisit the $p=1$ case on the way to $p=2$.

Here for convenience we rewrite the MIS problem Hamiltonian of Eq.~\ref{eq:C} as $C=c_0 I + \sum_j c_j Z_j + \sum_{\langle i,j\rangle} \lambda Z_iZ_j$, where $c_j=\lambda d_j -1$, and the $c_0 I$ term is irrelevant for our purposes and can be ignored. 
For clarity of presentation we use the notation $X_j:=\sigma_j^{(x)}$, $Y_j:=\sigma_j^{(y)}$, and $Z_j:=\sigma_j^{(z)}$ for the Pauli matrices, and $U_M=e^{-i\beta B}$, $U_P=e^{-i\gamma C}$ for the QAOA operators.

We calculate QAOA expectation values by considering the effect of its unitaries acting inwards by operator conjugation, as opposed to outwards in the state, and taking advantage of the Pauli operator commutation rules~\cite{hadfield2018quantum}. 
Consider the observable $\langle Z_j \rangle$. As terms in $B$ other than $X_j$ commute with both $Z_j$ and each other we have 
$$U_M^\dagger Z_j U_M = e^{-2 i \beta X_j} Z_j = \cos(2\beta)Z_j - \sin(2\beta)Y_j, $$
and so, using $C_j$ to denote the restriction to the terms in $C$ that contain a $Z_j$ factor, of which there are $d_j +1$ many, we similarly have 

\begin{widetext}
\begin{eqnarray}  \label{eq:p1}
    U_P^\dagger U_M^\dagger Z_j U_M U_P &=& \cos(2\beta)Z_j - \sin(2\beta)\,e^{2i \gamma C_j} Y_j\\
&=&\nonumber\cos(2\beta)Z_j - \sin(2\beta) \,e^{2i \gamma c_j Z_j }(\prod_{k\in nbhd(j)} e^{2i\gamma \lambda Z_jZ_k})Y_j\\
&=&\nonumber\cos(2\beta)Z_j - \sin(2\beta)\,(\cos (2 \gamma c_j) + i\sin(2\gamma c_j) Z_j)\left(\prod_{k\in nbhd(j)} (\cos (2 \gamma \lambda) + i\sin(2\gamma \lambda) Z_jZ_k)\right)Y_j.
\end{eqnarray} 
Taking initial state expectation values the first term on the right is zero and so does not contribute to the QAOA-1 expectation value $J_j(\gamma,\beta)=\langle Z_j\rangle_1$. For the second product term, it is easy to see the only subterm that can contribute is the product of the first sin factor with all remaining cos factors, and so plugging in the value of $c_j$ we have 
$$  \langle Z_j \rangle_1 =-\sin(2\beta) \sin(2\gamma (\lambda d_j -1)) \cos^{d_j}(2\gamma \lambda) $$
which reproduces Eq.~\ref{eq:Ji_final} as expected.

\paragraph*{Other $p=1$ expectation values:}
Next consider $\langle X_j\rangle_1$ and $\langle Y_j \rangle_1$ which turn out to be important quantities for analyzing the $p=2$ case. 
For $\langle Y_j\rangle_1$, repeating the above analysis we have 
$$U_M^\dagger Y_j U_M = e^{-2 i \beta X_j} Y_j = \cos(2\beta)Y_j + \sin(2\beta)Z_j, $$
and so 
$$U_P^\dagger U_M^\dagger Y_j U_M U_P = \cos(2\beta)e^{2i \gamma C_j} Y_j + \sin(2\beta)Z_j$$
which gives 
$$  \langle Y_j \rangle_1 =\cos(2\beta) \sin(2\gamma (\lambda d_j -1)) \cos^{d_j}(2\gamma \lambda). $$

Let $C'_j$ be $C_j$ restricted to the $ZZ$ terms. For $\langle X_j \rangle_1$,
the first mixing operator commutes through to similarly give 
$$\langle X_j \rangle_1 = \langle e^{2i\gamma C_j} X_j \rangle_0
=  \langle e^{2i\gamma C'_j} (\cos(2\gamma c_j) X_j  + \sin(2\gamma c_j) Y_j ) \rangle_0 = \cos(2\gamma c_j) \cos^{d_j}(2\gamma \lambda) = \cos(2\gamma (\lambda d_j-1)) \cos^{d_j}(2\gamma\lambda).$$

\paragraph*{p=2 case:}
Now turn to $\langle Z_j\rangle_2$. Observe that Eq.~\ref{eq:p1} implies that QAOA-$2$ expectation values can be written as linear combination of QAOA-$1$ operator expectation values 
\begin{eqnarray}\label{eq:Jp2}
\langle Z_j\rangle_2= \cos(2\beta_2)\,\langle Z_j\rangle_1(\gamma_1,\beta_1) -\sin(2\beta_2) \,\langle e^{2i\gamma_2 C_j} Y_j\rangle_1 (\gamma_1,\beta_1). 
\end{eqnarray}
The first term we get for free from the $p=1$ analysis above.
For the second term we have 
$$\langle e^{2i\gamma_2 C_j} Y_j\rangle_1 = \cos(2\gamma_2 c_j)\langle e^{2i\gamma_2 C'_j} Y_j\rangle_1 
+ \sin(2\gamma_2 c_j) \langle e^{2i\gamma_2 C'_j} X_j\rangle_1  $$

While this exact expression can be further expanded using either the path-sum or Pauli solver approaches, the number of contributing terms quickly becomes unwieldy, with different terms depending on different properties of the local graph neigborhood of vertex $j$. 
At the same time the form of Eq.~\ref{eq:Jp2} suggests utility of perturbative analysis~\cite{hadfield2022analytical} with respect to $\gamma_2$, which we briefly explore here. 

Assuming then relatively small $|\gamma_2|$, to leading order we  have 
$$\langle e^{2i\gamma_2 C_j}Y_j\rangle_1 = \cos(2\gamma_2 c_j) \cos^{d_j}(2\gamma_2 \lambda) \langle Y_j\rangle_1 
+ \sin(2\gamma_2 c_j) \cos^{d_j}(2\gamma_2 \lambda) \langle  X_j\rangle_1 + \dots,$$
where the terms not shown to the right contain at least two $\sin(\gamma_2 \cdot)$ factors. 
Applying the $p=1$ results above gives
\begin{eqnarray}\label{eq:p2}
\langle Z_j\rangle_2 &=& 
\cos(2\beta_2)\sin(2\beta_1) \sin(2\gamma_1 (1-\lambda d_j)) \cos^{d_j}(2\gamma_1 \lambda)\\
&-&\nonumber \sin(2\beta_2) \cos(2\gamma_2 c_j) \cos^{d_j}(2\gamma_2 \lambda) \cos(2\beta_1) \sin(2\gamma_1 (\lambda d_j -1)) \cos^{d_j}(2\gamma_1 \lambda) \\
&+& \nonumber  \sin(2\beta_2) \sin(2\gamma_2 c_j) \cos^{d_j}(2\gamma_2 \lambda) \cos(2\gamma_1 (\lambda d_j-1)) \cos^{d_j}(2\lambda\gamma_1)  
\,\, + \dots
\end{eqnarray}
The most important observation of Eq.~\ref{eq:p2} is that the first term is the same as the $p=1$ case, up to the factor of $\cos(2\beta_2)$. The following two terms then give \lq\lq correction factors,\rq\rq which once again only depend on the particular vetex degree $d_j$. Thus we see that to lowest order in $\gamma_2$ the $p=2$ case resembles a perturbed version of the $p=1$ case, i.e., a modified version of the classical MIN algorithm.
\end{widetext}

\section{Path-Sum on 2-Bit Correlators}
\label{app:2bit}

It is helpful to first consider 
\begin{align}
    \bra{z'}e^{i\beta \hat{B}}\sigma_i^{(z)}\sigma_j^{(z)}e^{-i\beta \hat{B}}\ket{z}=
    \\\nonumber
         \sum_{x\in\{-1,1\}^n}\bra{z'}e^{i\beta \hat{B}}\sigma_i^{(z)}\sigma_j^{(z)}\ket{x}\bra{x}e^{-i\beta \hat{B}}\ket{z}.
\end{align}
For the most part, this behaves the same as the object we looked at before.  The main difference is that now $\sigma_i^{(z)}\sigma_j^{(z)}\ket{x} = \ket{x'}$ where $x'$ is the same as $x$ except that both the $i$th and $j$th bits have been flipped.

Taking this into account, we get
\begin{align}
    \bra{z'}e^{i\beta \hat{B}}\sigma_i^{(z)}\sigma_j^{(z)}e^{-i\beta \hat{B}}\ket{z}=
    \\\nonumber
         \left(\prod_{k\neq i,j} \delta{z_k,z'_k}\right)\sum_{x_i,x_j = \pm1}e^{2i\beta (x_i+x_j)}
    \\\nonumber
    \times \braket{z_i'}{-x_i}\braket{x_i}{z_i}\braket{z_j'}{-x_j}\braket{x_j}{z_j}.
\end{align}
The brakets are fairly easy to evaluate using the rule mentioned above that
$\braket{x_j}{z_j} = \frac{1}{\sqrt{2}}(-1)^{(\frac{1+x_j}{2}) (\frac{1+z_j}{2})}$: 
\begin{align}
    \bra{z'}e^{i\beta \hat{B}}\sigma_i^{(z)}\sigma_j^{(z)}e^{-i\beta \hat{B}}\ket{z}=
    \\\nonumber
         \frac{1}{4}\left(\prod_{k\neq i,j} \delta{z_k,z'_k}\right)
    \\\nonumber
         \bigg(
           e^{-4i\beta}(-1)^{(\frac{1+z'_i}{2})+(\frac{1+z'_j}{2})}\\\nonumber
         + (-1)^{(\frac{1+z'_i}{2})+(\frac{1+z_j}{2})}\\\nonumber
         + (-1)^{(\frac{1+z_i}{2})+(\frac{1+z'_j}{2})}\\\nonumber
         + e^{4i\beta}(-1)^{(\frac{1+z_i}{2})+(\frac{1+z_j}{2})}
         \bigg).
\end{align}

Now, we can plug this into a formula for 
\begin{equation}
    J_{ij}(\beta,\gamma) = \bra{\varphi_0}e^{i\gamma \hat{C}}e^{i\beta \hat{B}}\sigma_i^{(z)}\sigma_j^{(z)}e^{-i\beta \hat{B}}e^{-i\gamma \hat{C}}\ket{\varphi_0}
\end{equation}
to get a simplified form summing over $z'_i=\pm z_i$ and $z'_j = \pm z_j$ 
\begin{align}
    J_{ij}(\beta,\gamma) &= \frac{1}{2^{n}}\sum_{z\in\{-1,1\}^n} (-1)^{(\frac{1+z_i}{2})+(\frac{1+z_j}{2})}\bigg[
    \\\nonumber
    &
    \cos^2(2\beta)\left(1+e^{-2i\gamma C^{(i,j)}(z)}\right)
    \\\nonumber
    &
    +\frac{i}{2}\sin(4\beta)\left(e^{-2i\gamma C^{(i)}(z)}+e^{-2i\gamma C^{(j)}(z)}\right)
    \bigg].
\end{align}
In this equation, we have introduced new notation.  $C^{(k)}(z)$ represents the terms in $C(z)$ that include a copy of $z_k$, and $C^{(k,m)}(z)$ represents the terms in $C(z)$ that include copies of either $z_k$ or $z_m$ but not both.

This is as far as we can go without specifying $C(z)$ more.  We could use this expression to calculate more with the Maximum Independent Set model above (such as calculating the expectation value of the entire energy function), or we could use it to carry out a similar procedure as above with different models such as MaxCut from the original RQAOA paper \cite{Bravyi2020}.

\paragraph*{Specifying to MIS}

We now specify down to Maximum Independent Set using $\hat{C}$ from Eq.~(\ref{eq:C}).  We furthermore will assume that the two qubits involved in this correlator have an edge between them in the graph structure.

After a lot of algebra that is very similar to what was discussed above, the end result for MIS is 
\begin{align}
    J_{ij}(\beta,\gamma) & =
        2\sin(4\beta)\sin(2\gamma\lambda)\cos(2\gamma(1-\lambda))\cos^{d_i}(2\gamma\lambda)\nonumber\\
        &+2\sin(4\beta)\sin(2\gamma\lambda)\cos(2\gamma(1-\lambda))\cos^{d_j}(2\gamma\lambda)\nonumber\\
        &+8\sin^2(2\beta)\cos^{d_i+d_j-2t_{ij}}(2\gamma\lambda)\cos(4\gamma(2\lambda-1))\nonumber\\
        &-8\sin^2(2\beta)\cos^{d_i+d_j}(2\gamma\lambda)\cos(4\gamma(2\lambda-1))\nonumber\\
        &+8\sin^2(2\beta)\cos^{d_i+d_j-2t_{ij}}(2\gamma\lambda).
\end{align}
Here we have introduced $t_{ij}$ to be the number of triangles that include the edge going between vertices $i$ and $j$.  Therefore, $d_i+d_j-2t_{ij}$ describes the number of vertices connected to $i$ or $j$ but not both.

\section{Correcting Dependent Sets for MIS} \label{app:MIScorr}

Here we describe a simple procedure that allows us to create valid Maximum Independent Set solutions given invalid ones at a small cost to the reward function. 
It can also be used to produce corrected observable expectation values. 
Consider the problem Hamiltonian formulation of Eq.~\ref{eq:C} that encodes the cost function $c(x)$ of Eq.~\ref{eq:c_classical}
with a penalty weight $\lambda\geq 1/2. $  Given any bit string~$x$ 
with cost function value $c(x)$, including penalties, we can obtain another solution which is feasible with cost at least $c(x)$.
The procedure is simple: at any step we randomly pick one of the penalized edges and remove one of its bits from the set, repeating until no penalized edges remain.  
This idea is also briefly observed for the $\lambda=1/2$ case in Ref.~\cite{farhi2020quantum_w}.
Hence, empirically estimated expectation values obtained from the quantum computer may still provide meaningful information to guide the algorithm 
even without post-selecting on feasible strings.

\begin{proposition} \label{prop:MIScorrection}
    For the objective function $c(x)=c_\lambda(x)$ of Eq.~\ref{eq:c_classical} with penalty weight $\lambda\geq 1/2$, for any $x\in\{0,1\}^n$ applying single bit-flip corrections yields an independent set $y$ of size at least $c(x)$.
\end{proposition}
\begin{corollary}
    For the objective Hamiltonian $\hat{C}$ of Eq.~\ref{eq:C} with $\lambda\geq 1/2$, and any $n$-qubit quantum state $\ket{\psi}$ which we measure in the computational basis, with high probability applying single bit-flip corrections to output strings produces an independent set $y$ with size at least $\langle \hat{C}\rangle_\psi$.
\end{corollary}

\begin{proof}
Recall 
$$c_\lambda(x)= -2\,r(x) + 4\,\lambda\, p(x) =-2\sum_{i=1}^n x_i + 4\,\lambda \sum_{\langle i, j\rangle} x_i x_j $$
where $r(x)$ is the reward function and $p(x)$ is the penalty.
Let $y\in \{0,1\}^n$ be a sample obtained from measuring the quantum state. Assume $y$ contains $\ell$ bits set to $1$. If $y$ is feasible we are done, $c(y)\geq c(y)$. Else there must exist at least one penalized edge with both vertices set to $1$. Select one of these vertices and flip it to $0$ to obtain string $y'$. Then $r(\cdot)$ will decrease by $1$, but $p(\cdot)$ will increase by at least $1$. Hence, with our chosen constants, if $\lambda\geq 1/2$ this step will never decrease the cost function $\hat{C}$. Note that the cost function strictly increases for the case $\lambda>1/2$.) After some number $L\leq n-1$ of steps the process terminates with an independent set $y$. 
As each step is non-decreasing we have that $y$ is of cardinality $c(y)\geq c(x)$ as claimed.

For the Corollary, measurements of the quantum state will produce a string with (penalized) cost at least $\langle \hat{C}\rangle_\psi$ with high probability after polynomially many repetitions~\cite{Farhi2014,hadfield2018quantum}, to which the proposition is applied. 
\end{proof}

We could also apply similar post-processing to mitigate the effect of infeasible states on QAOA expectation values. Similar ideas can be applied to other constrained optimization problems in particular cases.

\end{appendix}

\end{document}